\documentclass[final,3p]{elsarticle}
\usepackage{tikz}
\usepackage{pgfplots}
\pgfplotsset{compat=1.18}
\usepackage{graphicx} 
\usepackage{amssymb}
\usepackage{multirow} 
\usepackage{xcolor}
\usepackage{amsmath}
\usepackage{mathabx}
\usepackage{amsfonts}
\usepackage{amsthm}
\usepackage{tabularx}
\usepackage{slashed}
\usepackage{dcolumn}
\usepackage{accents}
\usepackage{array}
\usepackage{float}
\usepackage[T1]{fontenc}
\usepackage{graphicx}
\newtheorem{assumption}{Assumption}
\newcommand{\re}{\mathtt{r}}
\newcommand{\bl}{\mathtt{B}}
\newcommand{\cf}{c_{\mathrm{F}}}
\newcommand{\ci}{c_{\mathrm{I}}}
\newcommand{\setInitial}{\overline{\mathcal{I}}}
\newcommand{\setFinal}{\overline{\mathcal{F}}}
\newtheorem{theorem}{Theorem}
\newtheorem{lemma}{Lemma}
\newtheorem{corollary}{Corollary}
\newtheorem{remark}{Remark}
\begin{document}

\begin{frontmatter}



\title{Confidentiality in a Card-Based Protocol Under Repeated Biased Shuffles}


\author[inst]{Do Hyun Kim} 
\ead{am23105@shibaura-it.ac.jp}

\author[inst]{Ahmet Cetinkaya} 
\ead{ahmet@shibaura-it.ac.jp}

\affiliation[inst]{organization={College of Engineering, Shibaura Institute of Technology},
            addressline={Toyosu 3-7-5}, 
            city={Koto City},
            postcode={135-8448}, 
            state={Tokyo},
            country={Japan}}

\begin{abstract}
In this paper, we provide a probabilistic analysis of the confidentiality in a card-based protocol. We focus on Bert den Boer's original Five Card Trick to develop our approach. Five Card Trick was formulated as a secure two-party computation method, where two players use colored cards with identical backs to calculate the logical AND operation on the bits that they choose. In this method, the players first arrange the cards privately, and then shuffle them through a random cut. Finally, they reveal the shuffled arrangement to determine the result of the operation. An unbiased random cut is essential to prevent players from exposing their chosen bits to each other. However, players typically choose to move cards within the deck even though not moving any cards should be equally likely. This unconscious behavior results in a biased, nonuniform shuffling-distribution in the sense that some arrangements of cards are slightly more probable after the cut. Such a nonuniform distribution creates an opportunity for a malicious player to gain advantage in guessing the other player's choice. We provide the conditional probabilities of such guesses as a way to quantify the information leakage. Furthermore, we utilize the eigenstructure of a Markov chain to derive tight bounds on the number of times the biased random cuts must be repeated to reduce the leakage to an acceptable level. We also discuss the generalization of our approach to the setting where shuffling is conducted by a malicious player.
\end{abstract}



\begin{keyword}

Multi-Party computation\sep Card-based cryptography\sep Information leakage\sep Confidentiality\sep Probabilistic analysis\sep Markov chains



\end{keyword}

\end{frontmatter}

\section{Introduction}
As the value and utility of information continue to grow, ensuring the confidentiality of data has become increasingly important. Secure Multi-Party Computation (SMPC) is a key field in modern cryptography, enabling computation on inputs without revealing any information about them. Ideally, it should be impossible to deduce the inputs from the outputs~\cite{du2001secure}. However, performing SMPC can be challenging due to limited access to secure tools and the potential for malicious attacks on computational machines. An interesting idea discussed by Bert den Boer in \cite{den1990more} alleviates these issues in multi-party computation by using physical playing cards instead of relying on digital computation. This is the so-called Five Card Trick, which allows two players to calculate the result of the logical AND operation on the bits they each choose without revealing
their choice to the other player. The trick relies on private arrangement of cards by the two players and shuffling through
a random cut (bisection cut).

After the introduction of the original Five Card Trick, much research has been conducted on card-based protocols similar to the Five Card Trick, with a focus on reducing the number of the cards and shuffles required. Specifically, Mizuki and Sone~\cite{mizuki2009six} proposed a card-based protocol for computing the XOR operation, and later, Mizuki et al.~\cite{mizuki2012five} showed that AND operation could be performed using four cards instead of five. The paper~\cite{mizuki2014formalization} introduced a formal computational model for card-based protocols. Morever, Kastner et al.~\cite{kastner2017minimum} proved the minimum number of cards necessary for practical and secure card-based protocols implementing the AND and COPY operations. Beyond computational improvements, research on card-based protocols has explored operational concepts. In particular, a novel type of operation (referred to as a  ``private operation''~\cite{manabe2022card,morooka2023malicious,nakai2016efficient}) was introduced in the card-based protocol, allowing players to perform their actions without being observed by others. This increased the flexibility of card-based protocols by removing constraints on the number of cards needed for certain computations. However, it also brought additional vulnerabilities to card-based protocols that employ private operations. To address these concerns, Manabe et al.~\cite{manabe2022card} proposed methods to reinforce the confidentiality of private operations such as introducing a third party to monitor malicious behavior or using physical tools like envelopes to detect deviations from the protocol. 

In addition to reducing the number of cards, some studies have focused minimizing the number of shuffles required. In the original operation of the Five Card Trick, a shuffle method known as the \emph{random cut} was used. However, as various unique card-based protocols have been introduced, different shuffle techniques have also been adopted. These changes were generally aimed at reducing the number of shuffle iterations required. For example, Shinagawa et al.~\cite{shinagawa2015multi} introduced regular polygon cards to calculate the result of a function of non-binary inputs, while reducing the number of shuffles needed in card-based protocol. Furthermore, the paper~\cite{shinagawa2021single} proved that general secure computation can be achieved with only a single shuffle using their proposed protocols. Also, Honda and Shinagawa \cite{honda2024efficient} presented methods for computing AND and COPY operations efficiently in terms of both cards and the number of cuts.

To the best of our knowledge, previous papers do not handle the cases where shuffling is biased in the sense that some arrangement of cards are slightly more probable after shuffling. In this paper, we consider the scenario, where a player unintentionally introduces this bias, by being more likely to move some cards even though the choice of not shuffling should be equally likely. We use probability theory as our mathematical tool for analysis. To be more specific, we use conditional probability to show the potential confidentiality issues due to biased shuffling. Then we investigate what happens if the nonuniform shuffling process is repeated. To characterize repeated nonuniform shuffling, we use a Markov chain model. In our particular scenario, we observe that the Five Card Trick method does not entirely ensure confidentiality. We show that under certain cases, a malicious player can correctly guess the other player's input. On the other hand, if the nonuniform shuffles are repeated, confidentiality may improve. In this paper, we calculate a tight lower bound on the number of shuffles required to ensure a desired level of confidentiality. While our paper is concerned with confidentiality issues caused by the tendencies of biased shuffling, we also provide a discussion on a more general setting that allows us to handle the cases where one player is a malicious player and tries to make a certain order of cards more likely.

The use of Markov chains for modeling card shuffling has been considered previously by works such as \cite{chen2008cutoffrandomizedriffleshuffle,diaconis1996cutoff}, but with a theme different from card-based cryptography. Previous works mainly explore the so-called mixing property of Markov chains and the cut-off phenomenon, and they show that a surprisingly small number of ``riffle shuffles'' are sufficient to ensure that the order of cards are effectively randomized. Similar cut-off phenomena also exist in more a general setting of Markov processes \cite{chen2008cutoff}. Differently from past work, in this paper, we consider a confidentiality problem in a card-based protocol and explore random cuts instead of riffle shuffles. Furthermore, instead of assessing whether the card order is randomized, we analyze whether a player's bit-choice can be guessed by the other player after looking at the final order of cards. Our analysis technique also differs from those in \cite{chen2008cutoff,chen2008cutoffrandomizedriffleshuffle,diaconis1996cutoff} in that we do not directly investigate the mixing property of a Markov chain. Instead, we explore how a certain conditional probability related to information leakage evolves with respect to the number of shuffles.

We note that security aspects of the random cut has also been considered in the past work. Standard random cut is rather a simple method of shuffling the cards compared with the complicated implementation such as riffle shuffle. Therefore, because of its simplicity, there are chances that some players might track the number of the cards that moved \cite{ueda2020secure}. To mitigate this, Ueda et al.~\cite{ueda2016implement} proposed an alternative and secure implementation of a random cut. They pointed out that an aligned deck of cards and moving cards from bottom to the top when executing the random cut operations are more secure against the possible information leakage. Moreover, they showed that Hindu shuffle (Hindu cut) is an effective method, since it makes it much more difficult for the players to track the number of the cards moved in the operations. In this paper we focus on shuffling through a standard random cut; however, we believe that bias in other shuffling methods such as Hindu shuffling may be investigated in a similar fashion by using the conditional probability analysis that we present.

We remark that our analysis approach is applicable to other card-based protocols that use random cuts. In all card-based protocols random cuts may introduce bias, even though the number of cards may be different from five and the protocol may require extra operations. We decided to focus on the original five card trick protocol, because it is a standard in the literature and many other protocols are based on it.  

The organization of the remainder of this paper is as follows. In Section~\ref{Sec_2}, we summarize the original Five Card Trick and define notations for analysis. In Section~\ref{Sec_3}, we discuss confidentiality of the Five Card Trick and information leakage under biased shuffling. In Section~\ref{Sec_4}, we introduce a Markov chain model to characterize repeated shuffling and analyze the effects of repeated shuffles in reducing the information leakage. In Section~\ref{Sec_Generalization}, we discuss how we can adapt our analysis approach to a more general setting where there may be malicious shuffling. Finally, in Section~\ref{Sec_Conclusion}, we conclude our paper.

\section{Background and Notations}\label{Sec_2}
In this section, we provide a summary of the original Five Card Trick and introduce our notation for its analysis.
\subsection{The Five Card Trick}

In the Five Card Trick \cite{den1990more}, den Boer provides a way to securely compute the logical AND operation with five cards. There are two parties that participate in this calculation. In this paper, we identify these two parties as Alice and Bob. We consider the scenario that Alice and Bob want to calculate $a \land b$ where $a\in\{0,1\}$ is chosen by Alice and $b \in \{0,1\}$ is chosen by Bob. To do this calculation with privacy, the Five Card Trick uses three black cards and two red cards all with identical backside. In this paper, black cards and red cards will be represented with $\bl$ and $\re$, respectively.

To conduct the Five Card Trick, Alice and Bob are each given one pair of a black card and a red card. There is one extra Black card left to be used later. As a first step, Alice and Bob decide the order of their cards based on their bits as follows.
\begin{itemize}
\begin{item}[$\star$]
For Alice, $\re\bl$ means $a=1$, $\bl\re$ means $a=0$.
\end{item}
\begin{item}[$\star$]
For Bob, $\bl\re$ means $b=1$, $\re\bl$ means $b=0$.
\end{item}
\end{itemize}
After they make their decisions, they lay their cards facing down, in following the order:
Alice's cards -- the extra Black card -- Bob's cards. Then the cards are "shuffled" through a random cut. Finally, after shuffling, the final arrangement of the cards is revealed. If there are no three black cards adjacent to each other, and no two red cards adjacent to each other, then this means that the result of \( a \land b \) is 0. Otherwise it must be 1. 

The important privacy aspect of the Five Card Trick is that if one party chooses $0$ and the other party chooses $1$, then the party that chooses $0$ cannot determine the other party's choice by looking at the final arrangement of cards. 
For example, let's assume that Alice chooses $a=1$ with the resulting card order $\re\bl$,  and moreover, Bob chooses $b=0$ with the resulting card order $\re\bl$. In this case, the initial arrangement of cards will be $\re\bl\bl\re\bl$. 
After shuffling, if they have the final arrangement $\bl\re \bl\re\bl$, then this means that the result of the AND operation is $a\land b=0$.  In this case, while Alice can know from the result that Bob chose $b=0$, Bob cannot know what Alice chose and it can be either $a=1$ or $a=0$. This is because just by looking at the final arrangement $\bl\re\bl\re\bl$, Bob cannot guess whether the initial arrangement was $\re\bl\bl\re\bl$ or $\bl\re\bl\re\bl$, since both arrangements can result in the obtained final arrangement after a random cut. This example is illustrated in Fig.~\ref{fig1}.

\begin{figure}[t]
    \centering
    \includegraphics[width=0.8\linewidth]{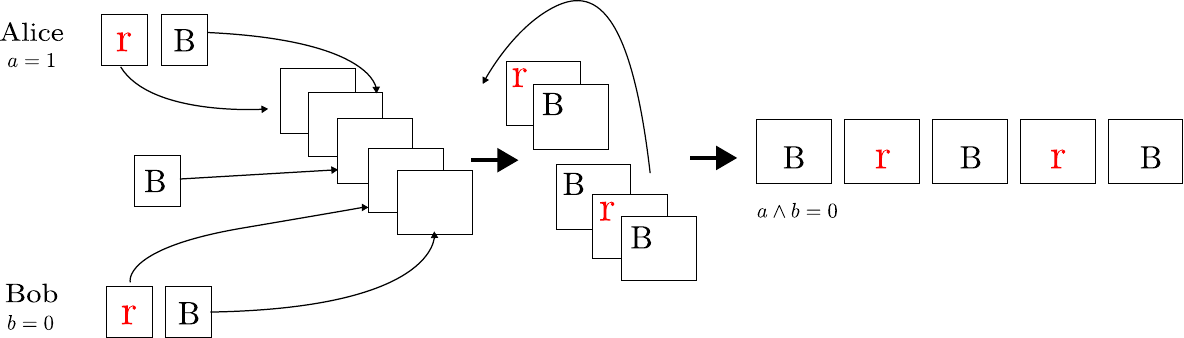}
    \caption{The process of the Five Card Trick}
    \label{The process of the Five Card Trick}
    \label{fig1}
\end{figure}

\subsection{Preliminary Notations}
In this paper, we explore confidentiality aspects of the Five Card Trick. To facilitate our analysis, we introduce several notations.  

\textbf{Sets of possible initial and final card arrangements:} We use $\mathcal{I}$ to denote the set of all possible initial arrangements, and $\mathcal{F}$ to denote the set of all possible final arrangements. Specifically,
\begin{align}
    \quad\mathcal{I} &= \{\re\bl\bl\re\bl, \bl\re\bl\re\bl, \bl\re\bl\bl\re, \re\bl\bl\bl\re\}, \\
    \quad\mathcal{F} &= \{\re\bl\bl\re\bl, \bl\re\bl\re\bl, \bl\re\bl\bl\re, \re\bl\bl\bl\re, \bl\bl\re\bl\re, \nonumber\\
    &\quad\quad \re\bl\re\bl\bl, \bl\bl\bl\re\re, \bl\bl\re\re\bl, \bl\re\re\bl\bl, \re\re\bl\bl\bl\}.
\end{align} {\vskip 3pt}

\textbf{Initial and final card arrangements:} The initial and the final arrangements of cards are defined respectively as random variables $\ci \colon \Omega \to \mathcal I$ and $\cf \colon \Omega \to \mathcal F$, where $\Omega$ is the set of outcomes in a probability space with probability measure $\mathbb{P}$. {\vskip 3pt}

\textbf{Shuffling order}: To model the number of cards moved from top to bottom after the random cut, we use the random variable $s \colon \Omega \to \{0,1,2,3,4\}$.  For instance, if a player cuts two cards from the top of the deck and moves them to the bottom of the deck, $s$ will be 2. Moreover, $s=0$ means that a player doesn't move any cards from the top. We also call $s$ the shuffling order. In den Boer's Five Card Trick, it is assumed that $s$ is uniformly distributed so that
\begin{align}
\mathbb{P}(s=i)=1/5, \quad i \in \{0,1,2,3,4\}. \label{s-equal-prob}
\end{align}
In other words, in the original Five Card Trick, all final arrangements are equally likely. Later, we will analyze the case where these probabilities are not uniform. {\vskip 3pt}

\textbf{Relationship between initial and final card arrangements:} To facilitate the analysis, we define $f \colon \mathcal{I} \times \{0,1,2,3,4\} \to \mathcal{F}$ as the function that determines the final arrangement of cards given an initial arrangement and the shuffling order. Given an initial arrangement $\mathtt{abcde}\in \mathcal{I}$, we have 
\begin{align*}
& f(\mathtt{abcde},0)=\mathtt{abcde},\quad f(\mathtt{abcde},1)=\mathtt{eabcd}, \\
& f(\mathtt{abcde},2)=\mathtt{deabc},\quad  f(\mathtt{abcde},3)=\mathtt{cdeab}, \\
& f(\mathtt{abcde},4)=\mathtt{bcdea}.
\end{align*}
For instance, $f(\re\bl\bl\re\bl,2)=\re\bl\re\bl\bl$.

As a result of all these notations that we defined, we have
\begin{align}
\cf(\omega)=f(\ci(\omega), s(\omega)) \label{def-cf}
\end{align}
for any outcome $\omega \in \Omega$. In the remainder of the paper, we omit specifying the outcome $\omega$, and write $\cf = f(\ci, s)$. 

\section{Confidentiality in the Five Card Trick Under a Biased Shuffle}\label{Sec_3}
In this section, we investigate the scenario in which a malicious player of the Five Card Trick gains an advantage in guessing the other player's choice by using a prior knowledge related to the probability distribution of the final arrangements of the cards after shuffling.

Shuffling through a random cut is a fundamental operation in card-based games and card-based protocols, typically assumed to enhance fairness. However, when players perform a random cut, they unconsciously force themselves to move some cards even though not moving any cards must be equally likely as moving $i>0$ number cards. This behavior is guided by the belief that it protects the player's confidentiality. In this section, we reveal that such unconscious behavior, when influenced by bias, may in fact lead to unintended information leakage and compromise the security of the protocol.

\subsection{The Five Card Trick Under a Biased Shuffle}
We now take a look at the Five Card Trick under the influence of biased shuffles where the bias is unintentionally introduced by one of the players (later in Section~\ref{Sec_Generalization} we will generalize this). Suppose that two players, Alice and Bob, are using the Five Card Trick protocol to calculate the result of AND operation on their selected bits, but Bob wants to guess the Alice's choice. For perfect confidentiality, the Five Card Trick requires a random  cut where the cut index is \emph{uniformly} distributed over $\{0,1,2,3,4\}$. If the cut index is $0$, then no cut is performed; when the cut index is $i>0$, then $i$ number of cards are moved from top to bottom. Bob knows that players are more likely to do a random cut with a non-zero cut index. Assume that Alice chooses the card order $\re\bl~(a=1)$ and Bob chooses the card order $\re\bl ~(b=0)$. The initial arrangement of the cards is $\re\bl\bl\re\bl$. Then, Bob asks Alice to shuffle the cards, and assume that we have $\re\bl\re\bl\bl$ as the final arrangement. Since players are cutting the deck of five cards (see Fig.~\ref{Exam_Card_Cuts}), there are five distinct possible outcomes. Therefore, each possible final arrangement should occur with probability $\frac{1}{5}$ under uniform randomness.
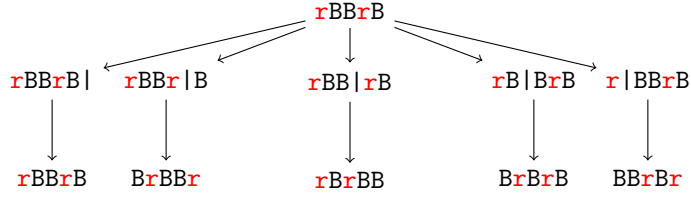
\begin{figure}[t]
\centering
\begin{tikzpicture}[node distance=1.3cm, every node/.style={align=center}]
    
    \node (start) {\texttt{\textcolor{red}{r}BB\textcolor{red}{r}B}};
    
    \node (cut1) [below left of=start, xshift=-3cm] {\texttt{\textcolor{red}{r}BB\textcolor{red}{r}B|}};
    \node (cut2) [below left of=start, xshift=-1.5cm] {\texttt{\textcolor{red}{r}BB\textcolor{red}{r}|B}};
    \node (cut3) [below of=start, yshift=0.35cm] {\texttt{\textcolor{red}{r}BB|\textcolor{red}{r}B}};
    \node (cut4) [below right of=start, xshift=1.5cm] {\texttt{\textcolor{red}{r}B|B\textcolor{red}{r}B}};
    \node (cut5) [below right of=start, xshift=3cm] {\texttt{\textcolor{red}{r}|BB\textcolor{red}{r}B}};

    \node (res1) [below of=cut1] {\texttt{\textcolor{red}{r}BB\textcolor{red}{r}B}};
    \node (res2) [below of=cut2] {\texttt{B\textcolor{red}{r}BB\textcolor{red}{r}}};
    \node (res3) [below of=cut3] {\texttt{\textcolor{red}{r}B\textcolor{red}{r}BB}};
    \node (res4) [below of=cut4] {\texttt{B\textcolor{red}{r}B\textcolor{red}{r}B}};
    \node (res5) [below of=cut5] {\texttt{BB\textcolor{red}{r}B\textcolor{red}{r}}};

    \draw[->] (start) -- (cut1);
    \draw[->] (start) -- (cut2);
    \draw[->] (start) -- (cut3);
    \draw[->] (start) -- (cut4);
    \draw[->] (start) -- (cut5);

    \draw[->] (cut1) -- (res1);
    \draw[->] (cut2) -- (res2);
    \draw[->] (cut3) -- (res3);
    \draw[->] (cut4) -- (res4);
    \draw[->] (cut5) -- (res5);
\end{tikzpicture}
\caption{Example of random cuts}\label{Exam_Card_Cuts}
\end{figure}

However, with Bob knowing the other player's behavioral characteristics, the probability of the first case in the Fig.~\ref{Exam_Card_Cuts} has a likelihood value slightly lower than $\frac{1}{5}$ while the other cases have slightly higher probability than $\frac{1}{5}$. This type of information leakage is difficult to detect and occurs naturally, as it does not require a malicious player to take a direct action in the process of the Five Card Trick. In what follows, we analyze how a malicious player can gain an advantage in guessing the other player's choice through a probabilistic approach.

\subsection{Analysis of the Biased Setting}
To describe how Bob can gain an advantage in guessing Alice's choice, we rely on a probabilistic analysis.   We begin by stating two assumptions that serve as a basis for our analysis.

\begin{assumption}\label{Assumption_equal_probability} Alice chooses either bit $0$ or $1$ with equal probability but Bob always chooses $0$, that is, 
\begin{align}
   &\mathbb{P}(a=1)=\mathbb{P}(a=0)=\frac{1}{2},\\
   &\mathbb{P}(b=0)=1.
\end{align}
\end{assumption}

\begin{assumption}\label{Assumption_nonuniform_prob} The shuffling order satisfies
 \begin{align}
      \quad\quad\mathbb{P}(s=0)& =\frac{1}{5}-\varepsilon,\label{s=0}\\ 
    \quad\quad\mathbb{P}(s=j) &=\frac{1}{5}+\frac{\varepsilon}{4}\quad\text{for}\,\,j\in \{1,2,3,4\}, \label{s=j} 
  \end{align}
where $\varepsilon\in [-\frac{4}{5},\frac{1}{5}]$. 
\end{assumption}

Under Assumption~\ref{Assumption_equal_probability}, since Bob's choice is fixed as $b=0$, the set of the initial arrangements is limited to two possible values as given by  
\begin{align}
&\quad  \setInitial\triangleq \{\re\bl\bl\re\bl,\bl\re\bl\re\bl\}. \label{nonuniform_initial_arrangements}
\end{align}
Corresponding to these initial arrangements, the set of the final arrangements is
\begin{align}
    &\quad \setFinal\triangleq \{\re\bl\bl\re\bl, \bl\re\bl\bl\re,\re\bl\re\bl\bl, \bl\re\bl\re\bl, \bl\bl\re\bl\re\}. \label{nonuniform_final_arrangements}
\end{align}

In Assumption~\ref{Assumption_nonuniform_prob}, we characterize the distribution of the shuffling order $s$, by using the parameter $\varepsilon$. While $\varepsilon$ can take values from the range $[-\frac{4}{5}, \frac{1}{5}]$,  in this section, we are interested in the case where $\varepsilon > 0$. With $\varepsilon > 0$, Assumption~\ref{Assumption_nonuniform_prob} implies that the probability of leaving the deck of cards in the initial state is lower than the probability of choosing to move a card. This assumption allows us to model the typical unconscious behavior of players who tend to do a cut with a non-zero cut index when asked to perform a random cut. The essential part of Assumption~\ref{Assumption_nonuniform_prob} in this paper is~\eqref{s=0}. Although it is possible to generalize~\eqref{s=j} so that the probability values are different, we use the setting with equal probabilities for simplicity of presentation. In our setting, the bias is characterized by the parameter $\varepsilon$. Larger values of $\varepsilon$ that are close to $\frac{1}{5}$ represent more drastic situations. We also note that negative values of $\varepsilon$ are shown to play a role in analysis in Section~\ref{Sec_4}.

\subsection{Confidentiality Analysis Using Conditional Probability}
Since the final arrangement of cards after shuffling is known, the security issue is whether this information can be used to infer about the initial arrangement. Conditional probability provides a framework directly related to this inference. To show that the Five Card Trick preserves confidentiality, we compute the conditional probability $\mathbb{P}(\ci=I \mid \cf=F)$ for a given initial arrangement $I\in \setInitial$ and the observed final arrangement $F \in \setFinal$. 

For example, consider the scenario that Alice chooses $a=1$ and Bob chooses $b=0$. The initial arrangement of cards will be $\re\bl\bl\re\bl$. Without loss of generality, let's further assume that after the shuffle, we have the arrangement $\bl\re\bl\re\bl$. We consider the situation that Bob wants to know Alice's choice. Since the final arrangement is known to Bob and Bob knows that he chose $b=0$ (Assumption\ref{Assumption_equal_probability}), the initial arrangement of cards must be either $\re\bl\bl\re\bl$ (indicating $a=1$) or $\bl\re\bl\re\bl$ (indicating $a=0$). In this case, we may be interested in calculating $\mathbb{P}(\ci=\re\bl\bl\re\bl \mid \cf=\bl\re\bl\re\bl)$. Here, 
\begin{align*}
\{\ci=\re\bl\bl\re\bl\}=\{\omega\in \Omega \colon \ci(\omega)=\re\bl\bl\re\bl\}
\end{align*}
denotes the event that the initial arrangement is $\re\bl\bl\re\bl$, and furthermore,
\begin{align*}
\{\cf=\bl\re\bl\re\bl\}=\{\omega\in \Omega \colon \cf(\omega)=\bl\re\bl\re\bl\}
\end{align*}
denotes the event that the final arrangement is $\bl\re\bl\re\bl$.

If $\mathbb{P}(\ci=\re\bl\bl\re\bl \mid \cf=\bl\re\bl\re\bl)>0.5$, it means that based on Bob's observation it is more likely that Alice chose $a=1$. If, on the other hand, $\mathbb{P}(\ci=\re\bl\bl\re\bl \mid \cf=\bl\re\bl\re\bl)<0.5$, then it is more likely that Alice chose $a=0$. Finally, if $\mathbb{P}(\ci=\re\bl\bl\re\bl \mid \cf=\bl\re\bl\re\bl)=0.5$, then Bob's observations do not help him guess Alice's bit, since $a=1$ and $a=0$ are equally likely.

We are now ready to present our main result that fully characterizes the conditional probability $\mathbb{P}(\ci=I \mid \cf=F)$. 

\begin{theorem}\label{theorem_cond_prob}
If Assumptions~\ref{Assumption_equal_probability} and \ref{Assumption_nonuniform_prob} both hold, then the conditional probability $\mathbb{P}(\ci=I \mid \cf=F)$ is characterized by the following two cases. \vspace{0.3cm}\newline 
$Case_1$: For $F \in \{f(I,0) : I \in \setInitial\}$ and $I \in \setInitial$,
\begin{align}
  \mathbb{P}(\ci=I \mid \cf=F) =
    \begin{cases}
        \frac{4-20\varepsilon}{8-15\varepsilon}, &\text{if}\,\, F=f(I,0) \vspace{0.3cm}\\
        \frac{4+5\varepsilon}{8-15\varepsilon},  &\text{otherwise.}
    \end{cases} \label{theorem1-result-case1}
\end{align}
$Case_2$: For $F \notin \{f(I,0) : I \in \setInitial\}$ and $I \in \setInitial$,
\begin{align}
      \mathbb{P}(\ci=I \mid \cf=F) = \frac{1}{2}.  \label{theorem1-result-case2}
\end{align}
\end{theorem}

Proof of Theorem~\ref{theorem_cond_prob} is presented later in Section~\ref{sec:proof}. 

Theorem~\ref{theorem_cond_prob} addresses two different cases, $Case_1$ wherein the two-party computation proves insecure, and $Case_2$ where it remains secure. The final arrangements in $Case_1$ are characterized as $F \in \{f(I,0) : I \in \setInitial\}$. This is the case where $F\in\{\re\bl\bl\re\bl,\bl\re\bl\re\bl\}$. On the other hand, in $Case_2$, $F\in\{\re\bl\re\bl\bl, \bl\re\bl\bl\re, \bl\bl\re\bl\re\}$.

\subsubsection{Discussion on $Case_1$:}
Notice that in $Case_1$, the conditional probability has separate outcomes for two distinct situations that are determined based on the final arrangements of the cards.
In the first situation ($F=f(I,0)$), we are interested in the conditional probability of having a particular initial arrangement $I$, given that the final arrangement is the same as that arrangement (i.e., $F=f(I, 0)=I$). In that situation, if $\varepsilon > 0$, we have 
\begin{align}
    \mathbb{P}(\ci=I \mid \cf=F) = \frac{4-20\varepsilon}{8-15\varepsilon} < \frac{1}{2}. \label{situation-1}
\end{align}
In the second situation ($F\neq f(I,0)$), we are interested in the conditional probability when $F\neq I$. In that situation, if $\varepsilon>0$, we have
\begin{align}
    \mathbb{P}(\ci=I \mid \cf=F) = \frac{4+5\varepsilon}{8-15\varepsilon} > \frac{1}{2}. \label{situation-2}
\end{align}
Notice that \eqref{situation-1} and \eqref{situation-2} imply that in $Case_1$, when $\varepsilon>0$, the conditional probability is always different from $0.5$ (see Fig.~\ref{Theorem1_case1_figure}
for the values of conditional probabilities with respect to different values of $\varepsilon$).  
As a result, the malicious player (Bob) can indeed gain advantage in guessing the other player's choice. 

\begin{figure}[t]
    \centering
    \includegraphics[width=1\linewidth]{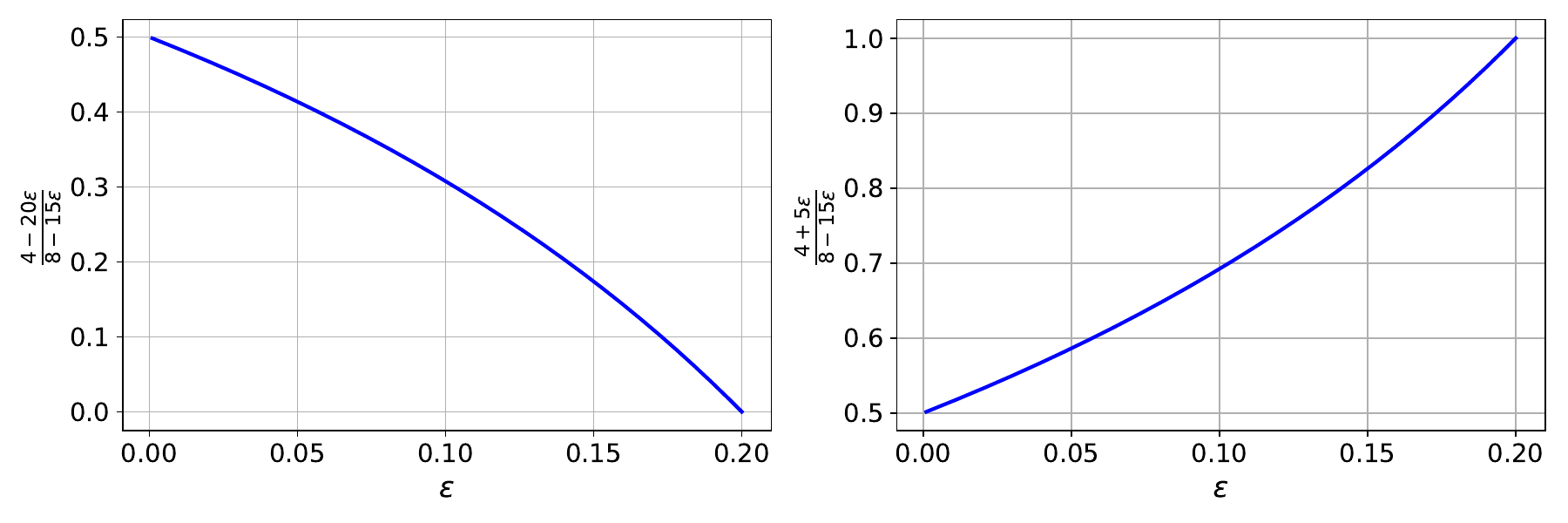}
    \caption{Conditional probabilities in $Case_1$ for nonnegative values of $\varepsilon$}
    \label{Theorem1_case1_figure}
\end{figure}

In particular, for the observed final arrangement $F$, if $F\in\{f(I,0) : I \in \setInitial\}$ (i.e., $Case_1$), Bob can check the values of $\mathbb{P}(\ci=I \mid \cf=F)$ for the two possible initial arrangements $I=\re\bl\bl\re\bl$ and $I=\bl\re\bl\re\bl$ and see which one is larger.

This demonstrates that Bob can make informed guesses about Alice’s choice, relying exclusively on the final arrangement of the cards by exploiting the bias in $\mathbb{P}(s = 0) = \frac{1}{5} - \varepsilon$. As $\varepsilon$ approaches $\frac{1}{5}$, the shuffling player (Alice) becomes progressively less likely to leave the cards unchanged. In the limiting scenario where $\varepsilon = \frac{1}{5}$, the unshuffled scenario no longer occurs. From Bob’s perspective, this enhances exploitable information as the deviation of the conditional probability from $\frac{1}{2}$ becomes more significant. The value of $\varepsilon$ is likely influenced by the behavioral tendencies of Alice performing the shuffle. However, knowing the exact value of $\varepsilon$ is not always possible for Bob.  Depending on the level of information available to Bob about the bias, Bob can be more accurate in his guesses. We illustrate this through the 3 information levels presented below.
\begin{itemize}
    \begin{item} [1)] If Bob knows the existence of a positive parameter $\varepsilon$, but does not know the exact value of it, then he can only know whether $\mathbb{P}(\ci=I \mid \cf=F)$ is larger than or smaller than $0.5$. However, the mere fact that $\varepsilon$ is positive already gives Bob with non-negligible information. \vskip 5pt
    \end{item}
    \begin{item} [2)] If Bob knows a positive lower-bound $\underline{\varepsilon}\in(0,\frac{1}{2}]$ such that $\varepsilon \geq \underline{\varepsilon}$, then Bob can have a better understanding of the conditional probability compared to Level 1 above. In particular, Bob can obtain the bounds $\mathbb{P}(\ci=I \mid \cf=F)\leq \frac{4-20\underline{\varepsilon}}{8-15\underline{\varepsilon}}$ for $F=I$, and $\mathbb{P}(\ci=I \mid \cf=F) \geq \frac{4+5\underline{\varepsilon}}{8-15\underline{\varepsilon}}$ for $F\neq I$.  \vskip 5pt
    \end{item}
     \begin{item} [3)] If Bob knows $\varepsilon$ exactly (e.g., by using data from past observations), then Bob can compute $\mathbb{P}(\ci=I \mid \cf=F)$ exactly. 
    \end{item}
\end{itemize}

\subsubsection{Discussion on $Case_2$:}
$Case_2$ of the Theorem~\ref{theorem_cond_prob} represents the scenario where the two-party computation remains secure. If the final arrangement does not match any of the unshuffled forms (i.e., $F \notin \{f(I,0) : I \in \setInitial\}$), the conditional probability remains exactly $0.5$ for all inputs. In this situation, the malicious player cannot infer any information about the other party’s choice unless they actively manipulate the output space.

To conclude, even though $Case_2$ shows that there is no information leakage, the behavioral tendencies in card shuffling can lead to information leakage and pose security risks in $Case_1$. A malicious player may use this information to threaten the confidentiality of the other player's information, which makes perfectly secure multi-party computation difficult to achieve.

\subsection{Proof of Theorem~\ref{theorem_cond_prob}}
\label{sec:proof}
The proof of Theorem~\ref{theorem_cond_prob} relies on the following three lemmas. Their proofs are presented in the Appendix.
\begin{lemma} \label{Lemma-s-star}  For any $r\in\{0,1,2,3,4\}$ and $I\in \setInitial$, we have
\begin{align}
   \quad\quad \mathbb{P}\left( f(I,s)=f(I,r)\right) = \mathbb{P}(s=r). \label{lem1-result}
\end{align}
\end{lemma}

\begin{lemma} \label{Lemma-s-i} Suppose Assumption~\ref{Assumption_nonuniform_prob} holds. Then for any given  $i, j\in \setInitial$, we have
\begin{align}
\quad\quad\mathbb{P}( f(i,s)=f(j,0)) =
\begin{cases}
    \frac{1}{5}-\varepsilon, & \text{if}\,\, i=j,\vspace{0.2cm} \\
    \frac{1}{5}+\frac{\varepsilon}{4}, & \text{if}\,\, i\neq j.
\end{cases} \label{s-i-result}
\end{align}
\end{lemma}

\begin{lemma}\label{lemma_total_prob}
    Suppose Assumption~\ref{Assumption_equal_probability} and ~\ref{Assumption_nonuniform_prob} hold. Then for any given final arrangement $F \in \{f(I,0) : I \in \setInitial\}$, we have
    \begin{align}
        \mathbb{P}(\cf = F) = \frac{1}{2}(\frac{2}{5}-\frac{3\varepsilon}{4}).\label{lemma_total_prob_eq}
    \end{align}
\end{lemma}

We are now ready to prove Theorem~\ref{theorem_cond_prob} by using Lemmas~\ref{Lemma-s-star}--\ref{lemma_total_prob}.
\begin{proof}[Theorem~\ref{theorem_cond_prob}]
By the definition of conditional probability and \eqref{def-cf},
\begin{align}
    \mathbb{P}(\ci=I \mid \cf=F) &= \frac{\mathbb{P}(\ci=I,\cf=F)}{\mathbb{P}(\cf=F)} = \frac{\mathbb{P}(\ci=I, f(\ci, s)=F)}{\mathbb{P}(\cf=F)}\nonumber\nonumber \\ 
  &= \frac{\mathbb{P}(\ci=I, f(I, s)=F)}{\mathbb{P}(\cf=F)}. \label{cond-prob-result}
\end{align}

Now, since $\ci$ and $s$ are independent and $\mathbb{P}(\ci=I)=1/2$ (by Assumption~\ref{Assumption_equal_probability}), we obtain from \eqref{cond-prob-result} that 
\begin{align}
    \mathbb{P}(\ci=I \mid \cf=F) &= \frac{\mathbb{P}(\ci=I)\mathbb{P}(f(I, s)=F)}{\mathbb{P}(\cf=F)}=\frac{\mathbb{P}(f(I, s)=F)}{2\mathbb{P}(\cf=F)}. \label{before-case-separation}
\end{align}

Next, we use \eqref{before-case-separation} to prove \eqref{theorem1-result-case1} for $Case_1$ and \eqref{theorem1-result-case2} for $Case_2$. 

$Case_1$: Since $F \in \{f(I,0) : I \in \setInitial\}$, it follows from Lemma~\ref{lemma_total_prob} that $\mathbb{P}(\cf=F)=\frac{1}{2}(\frac{2}{5}-\frac{3\varepsilon}{4})$. Therefore, \eqref{before-case-separation} yields 
\begin{align}
    \mathbb{P}(\ci=I \mid \cf=F) &= \frac{\mathbb{P}(\ci=I)\mathbb{P}(f(I, s)=F)}{\mathbb{P}(\cf=F)}=\frac{\mathbb{P}(f(I, s)=F)}{\frac{2}{5}-\frac{3\varepsilon}{4}}. \label{before-case-separation_case1}
\end{align}

Next we evaluate $\mathbb{P}(f(I, s)=F)$. Since $F \in \{f(I,0) : I \in \setInitial\}$ and $I \in \setInitial$, we have $I,F\in\setInitial$. In other words, both $I$ and $F$ have two possible values, either $\re\bl\bl\re\bl$ or $\bl\re\bl\re\bl$. Given $I\in \setInitial$, let $\widecheck{I} \in \setInitial$ denote the arrangement such that $\{I,\widecheck{I}\}=\setInitial$. We consider two situations: 1) $F=I$ and 2) $F\neq I$. 

If $F=I$, then we have 
\begin{align}
   \mathbb{P}(f(I, s)=F)=\mathbb{P}(f(I, s)=I)=\mathbb{P}(f(I, s)=f(I, 0)). \label{s1-1} 
\end{align}
Here, by Lemma~\ref{Lemma-s-i} with $i=I$, we obtain $\mathbb{P}(f(I, s)=f(I, 0))=1/5-\varepsilon$. Therefore, it follows from \eqref{s1-1} that
\begin{align}
    \mathbb{P}(f(I, s)=F) = \frac{1}{5}-\varepsilon, \label{situtation1-result}
\end{align}
and thus \eqref{before-case-separation_case1} yields 
\begin{align}
   \mathbb{P}(f(I, s)=F)=\frac{\frac{1}{5}-\varepsilon}{\frac{2}{5}-\frac{3\varepsilon}{4}}=\frac{4-20\varepsilon}{8-15\varepsilon}. \label{situation1-result-final}
\end{align}

If $F\neq I$, then it means that $F=\widecheck{I}$. Hence, 
\begin{align}
   \mathbb{P}(f(I, s)=F)=\mathbb{P}(f(I, s)=\widecheck{I})=\mathbb{P}(f(I, s)=f(\widecheck{I}, 0)). \label{s2-1} 
\end{align}
Here, by Lemma~\ref{Lemma-s-i} with $i=I$ and $j=\widecheck{I}$, we obtain $\mathbb{P}(f(I, s)=f(\widecheck{I}, 0))=1/5+\varepsilon/4$. Therefore, it follows from \eqref{s2-1} that
\begin{align}
    \mathbb{P}(f(I, s)=F) = \frac{1}{5}+\frac{\varepsilon}{4}, \label{situtation2-result}
\end{align}
and thus \eqref{before-case-separation_case1} yields 
\begin{align}
   \mathbb{P}(f(I, s)=F)=\frac{\frac{1}{5}+\frac{\varepsilon}{4}}{\frac{2}{5}-\frac{3\varepsilon}{4}}=\frac{4+15\varepsilon}{8-15\varepsilon}.  \label{situation2-result-final}
\end{align}
By combining \eqref{situation1-result-final} for $F=f(I,0)=I$ and \eqref{situation2-result-final} for $F\neq f(I, 0)$, we get \eqref{theorem1-result-case1}. \vskip 3pt

 $Case_2$: Notice that in this case, we have $F \notin \{f(I,0) : I \in \setInitial\}$ and $I \in \setInitial$.
Therefore, $F\in \{\bl\re\bl\bl\re,\re\bl\re\bl\bl, \bl\bl\re\bl\re\}, $ $I\in\{\re\bl\bl\re\bl, \bl\re\bl\re\bl\}$, which implies $F\neq I$. Let $q_{I,F} \in \{1,2,3,4\}$ denote the index such that $F=f(I,q_{I,F})$ (for instance for $I=\re\bl\bl\re\bl$ and $F=\bl\bl\re\bl\re$, we have $q_{I,F}=4$). 

For given $I\in \{\re\bl\bl\re\bl, \bl\re\bl\re\bl\}$, let $\widecheck{I}$ denote the arrangement such that $\{I, \widecheck{I}\}=\{\re\bl\bl\re\bl, \bl\re\bl\re\bl\}$. Notice also that $F\neq \widecheck{I}$. Similarly to $q_{I,F}$, let $q_{\widecheck{I},F} \in \{1,2,3,4\}$ denote the index such that $F=f(\widecheck{I},q_{\widecheck{I},F})$.

By Law of Total Probability, we write
\begin{align}
\mathbb{P}(\cf=F) &= \mathbb{P}(\ci=I, \cf=F)+\mathbb{P}(\ci=\widecheck{I}, \cf=F). \label{case2-sum-of-probabilities}
\end{align}
By using Assumption~\ref{Assumption_equal_probability},  Lemma~\ref{Lemma-s-star}, as well as Assumption~\ref{Assumption_nonuniform_prob}, we have
\begin{align}
    \mathbb{P}(\ci=I, \cf=F)&=\mathbb{P}(\ci=I)\mathbb{P}(f(I,s)=F)=\frac{1}{2}\mathbb{P}(f(I,s)=f(I,q_{I,F})) \nonumber \\
    &=\frac{1}{2}\mathbb{P}(s=q_{I,F})=\frac{1}{2}\Big(\frac{1}{5}+\frac{\varepsilon}{4}\Big). \label{case2-equality1}
\end{align}
Similarly, we can compute 
\begin{align}
    \mathbb{P}(\ci=\widecheck{I}, \cf=F)&=\mathbb{P}(\ci=\widecheck{I})\mathbb{P}(f(\widecheck{I},s)=F)=\frac{1}{2}\mathbb{P}(f(\widecheck{I},s)=f(\widecheck{I},q_{\widecheck{I},F})) \nonumber \\&=\frac{1}{2}\mathbb{P}(s=q_{\widecheck{I},F})=\frac{1}{2}\Big(\frac{1}{5}+\frac{\varepsilon}{4}\Big). \label{case2-equality2}
\end{align}

Notice that \eqref{case2-equality1} and \eqref{case2-equality2} imply  $\mathbb{P}(\ci=I, \cf=F)=\mathbb{P}(\ci=\widecheck{I}, \cf=F)$. Thus, \eqref{case2-sum-of-probabilities}
yields 
\begin{align}
\mathbb{P}(\cf=F) &= 2\mathbb{P}(\ci=I, \cf=F) = 2\mathbb{P}(\ci=I)\mathbb{P}(\cf=F)  \nonumber \\ 
 & = 2\mathbb{P}(\ci=I)\mathbb{P}(f(I, s)=F) = \mathbb{P}(f(I, s)=F). \label{Case2_total_prob}
\end{align}
Substituting the identity derived in~\eqref{Case2_total_prob} into \eqref{before-case-separation}, we obtain \eqref{theorem1-result-case2}.
\hfill $\square$ \end{proof}

\section{Repeated Random Cuts for Security}\label{Sec_4}
As we studied in the previous sections, the process of the random cut is closely related to the security aspects of the Five Card Trick. In this section, we investigate \emph{repeated} random cuts as a potential method of improving confidentiality, even if those cuts are biased. Specifically, we use a Markov chain~\cite{diaconis1996cutoff,norris1998markov} to characterize the repeated random cuts. 

\subsection{Characterization of repeated cuts through a Markov chain}
In Section~\ref{Sec_2}, we used the random variable $s\colon\Omega \to \{0,1,2,3,4\}$ to denote shuffling order, i.e., the order of the final arrangement of cards after a cut. To characterize repeated random cuts, we now use a Markov chain $\{r(t)\in\{0,1,2,3,4\}\}_{t\in{\mathbb{N}_0}}$ with initial distribution vector $\nu\in \mathbb{R}^{1\times 5}$ and transition probability matrix $P\in\mathbb{R}^{5\times 5}$. In this characterization, for a given nonnegative integer $t\in \mathbb{N}_0$, the random variable $r(t)\colon\Omega\to\{0,1,2,3,4\}$ denotes the order of the final arrangement of cards after $t$ number of random cuts. To simplify derivations that involve $\nu$ and $P$, we use the notion that the entries of vectors and matrices start from $0$, and thus, we have
\begin{align}
\mathbb{P}(r(0)=i) &=\nu_i, \quad i\in\{0,1,2,3,4\}, \label{markov-1} \\
\mathbb{P}(r(t+1)=j|r(t)=i) &=P_{i,j}, \quad i,j \in\{0,1,2,3,4\} \label{markov-2}.
\end{align}


Similar to the setting in Section~\ref{Sec_3}, we want to consider \emph{biased} random cuts, where cutting at the zero cut-index has a different probability than other cut indices. By taking into account the cyclic nature of cuts, the bias is characterized by the transition probability matrix 
\begin{align}
P = \begin{bmatrix}
a & b & b & b & b \\
b & a & b & b & b \\
b & b & a & b & b \\
b & b & b & a & b \\
b & b & b & b & a
\end{bmatrix}, \label{P-matrix}
\end{align}
where $a\in[0,1]$ denotes the probability of cutting at the zero cut-index, and $b=\frac{1-a}{4}$. In the setting of Assumption~\ref{Assumption_nonuniform_prob}, $a=\frac{1}{5}-\varepsilon, b=\frac{1}{5}+\frac{\varepsilon}{4}$. Again, for unintentional bias introduced by being more likely to do a cut from a nonzero index is handled by setting $\varepsilon > 0$. In such a case, we have $a < b$. However, our results in this section also cover the case where $a>b$. 
We note that the zero cut-index depends on the order of the current arrangement of the cards. As a result, the probability $a$ appears on the diagonal, since starting from the $i$th arrangement order, a cut with zero cut-index results again in the same arrangement order $i$.  

Furthermore, the initial distribution vector $\nu$ is set as 
\begin{align}
\nu = \begin{bmatrix}
1 & 0 & 0 & 0 & 0 
\end{bmatrix}, \label{nu-vector}
\end{align}
so that the initial order of arrangement is 0, that is, $r(0)=0$. 

The following lemma provides the probabilities regarding the possible orders of arrangements of cards after $t$ number of random cuts. Its proof relies on the eigenstructure of the matrix $P$ in \eqref{P-matrix}.

\begin{lemma} 
\label{Lemma-rt}
The Markov chain $\{r(t)\}_{t\in\mathbb{N}_0}$ with transition probability matrix $P$ in \eqref{P-matrix} and initial distribution vector $\nu$ in \eqref{nu-vector} satisfies
\begin{align}
\mathbb{P}(r(t)=i)=\begin{cases}
\frac{1}{5}+ \frac{4}{5}(a - b)^t, &\quad i=0,\\
\frac{1}{5}- \frac{1}{5}(a - b)^t, &\quad i\in\{1,2,3,4\}.
\end{cases} \label{pst-i-result}
\end{align}
\end{lemma}
\begin{proof} First, it follows from \eqref{markov-1} and \eqref{markov-2} that 
\begin{align}
    \mathbb{P}(r(t)=i) = (\nu P^t)_i, \quad i\in\{0,1,2,3,4\}. \label{pst-i}
\end{align}
To evaluate \eqref{pst-i}, we need to compute $P^t$. To this end, we first analyze the eigenstructure of $P$. Note that $P$ can be written as
\begin{align}
    P = (a - b)I + bJ, \label{PIJ}
\end{align}
where $I\in \mathbb{R}^{5\times5}$ is the identity matrix and $J\in\mathbb{R}^{5\times5}$ is the matrix with all entries equal to $1$. Since the eigenvalues of the matrix $J$ are $l_1=0$ and $l_2=l_3=l_4=l_5=5$, it follows from \eqref{PIJ} that the eigenvalues of $P$ can be computed using the identity $\lambda_i=(a-b)+bl_i$ as 
\begin{align}
   \lambda_1 = a + 4b, \quad \lambda_2 = \lambda_3 = \lambda_4 = \lambda_5 = a - b. 
\end{align}
The right-eigenvectors corresponding to these eigenvalues are
\begin{align}
    v_1 = \begin{bmatrix} 1 \\ 1 \\ 1 \\ 1 \\ 1 \end{bmatrix}, \quad
v_2 = \begin{bmatrix} -1 \\ \phantom{-}1 \\ \phantom{-}0 \\ \phantom{-}0 \\ \phantom{-}0 \end{bmatrix}, \quad
v_3 = \begin{bmatrix} -1 \\ \phantom{-}0 \\ \phantom{-}1 \\ \phantom{-}0 \\ \phantom{-}0 \end{bmatrix}, \quad
v_4 = \begin{bmatrix} -1 \\ \phantom{-}0 \\ \phantom{-}0 \\ \phantom{-}1 \\ \phantom{-}0 \end{bmatrix}, \quad
v_5 = \begin{bmatrix} -1 \\ \phantom{-}0 \\ \phantom{-}0 \\ \phantom{-}0 \\ \phantom{-}1 \end{bmatrix}. \label{right_eigen}
\end{align}
Consider the matrix $T\in \mathbb{R}^{5\times 5}$ formed as $T=[v_1,\,v_2,\,v_3,\,v_4,\,v_5]$. Noting that the eigenvectors $v_i$ are linearly independent (and thus $P$ is diagonalizable), it follows by similarity transformation that $T^{-1}PT=\mathrm{diag}(\lambda_1, \lambda_2, \lambda_3, \lambda_4, \lambda_5)$. Therefore, $P=T\mathrm{diag}(\lambda_1, \lambda_2, \lambda_3, \lambda_4, \lambda_5)T^{-1}$ and consequently, 
\begin{align}
    P^t &= T\mathrm{diag}(\lambda_1^t, \lambda_2^t, \lambda_3^t, \lambda_4^t, \lambda_5^t)T^{-1}. \label{Pt}
\end{align}
It follows from \eqref{Pt} by direct computation that
\begin{align}
    (P^t)_{i,j} =
    \begin{cases}
    \frac{1}{5}(a + 4b)^t + \frac{4}{5}(a - b)^t, & \text{if } i = j, \\
    \frac{1}{5}(a + 4b)^t - \frac{1}{5}(a - b)^t, & \text{if } i \neq j.
    \end{cases} \label{Ptij}
\end{align}
Noting that $a+4b=1$, we use \eqref{nu-vector}, \eqref{pst-i}, and \eqref{Ptij} to obtain \eqref{pst-i-result}. \hfill $\square$
\end{proof}

\subsection{Tight Bound on Required Number of Shuffles}
In what follows, we obtain tight bounds on the number of shuffles (i.e., the number of cuts) that is required to keep confidentiality at a desired level. As a first result, we compute the conditional probability $\mathbb{P}(\ci=I\mid\cf=F)$ for different values of the number of shuffles.

\begin{theorem} 
\label{Theorem-repeated}
Suppose Assumption~\ref{Assumption_equal_probability} holds.
    After $T\in \mathbb{N}_0$ number of repeated biased shuffles characterized through the Markov chain $\{r(t)\}_{t\in \mathbb N_0}$, the conditional probability $\mathbb{P}(\ci=I\mid \cf =F)$ is characterized as follows. \vspace{0.3cm}\newline 
    $Case_1$: For $F \in \{f(I,0) : I \in \setInitial\}$ and $I \in \setInitial$,
\begin{align}
  \mathbb{P}(\ci=I \mid \cf=F) =
    \begin{cases}
        \frac{4 + 16(a-b)^T}{8+12(a-b)^T}, &\text{if}\,\, F=f(I,0) \vspace{0.3cm}\\
        \frac{4-4(a-b)^T}{8+12(a-b)^T},  &\text{otherwise.}
    \end{cases} \label{theorem2-result-case1}
\end{align}
$Case_2$: For $F \notin \{f(I,0) : I \in \setInitial\}$ and $I \in \setInitial$,
\begin{align}
      \mathbb{P}(\ci=I \mid \cf=F) = \frac{1}{2}.  \label{theorem2-result-case2}
\end{align}
\end{theorem}
\begin{proof}
Let $s\triangleq r(T)$. By Lemma~\ref{Lemma-rt}, Assumption~\ref{Assumption_nonuniform_prob} holds with 
    \begin{align}
        \varepsilon=- \frac{4}{5}(a - b)^T.
    \end{align}
The result then follows from Theorem~\ref{theorem_cond_prob} with this $\varepsilon$ value. \hfill $\square$\end{proof}

As shown in the proof Theorem~\ref{Theorem-repeated}, under repeated shuffling, Assumption~\ref{Assumption_nonuniform_prob} holds with $\varepsilon=- \frac{4}{5}(a - b)^T$.

\begin{remark}[Parity of the number of shuffles] Theorem~\ref{Theorem-repeated} indicates that the conditional probability $\mathbb{P}(\ci=I\mid\cf=F)$ in $Case_1$ depends on the number of shuffles $T$. In \eqref{theorem2-result-case1}, we observe that if $a < b$ (as in the case of unintentional bias in random cut discussed earlier), then $\frac{4 + 16(a-b)^T}{8+12(a-b)^T} < \frac{4-4(a-b)^T}{8+12(a-b)^T}$ if $T$ is odd, and moreover, $\frac{4 + 16(a-b)^T}{8+12(a-b)^T} > \frac{4-4(a-b)^T}{8+12(a-b)^T}$ if $T$ is even. This means that the conditional probabilities in both situations $F=f(I,0)$ and $F\neq f(I,0)$ oscillate around $\frac{1}{2}$ depending on the parity $T$. Thus, if Bob does not know the parity of $T$, it confuses Bob in inferring Alice's choice. This is because in $Case_1$, Bob cannot figure out whether $\mathbb{P}(\ci=I \mid \cf=F) > \frac{1}{2}$ or $\mathbb{P}(\ci=I \mid \cf=F) < \frac{1}{2}$, since either can be true depending on the parity of $T$. \hfill $\triangleleft$
\end{remark}

If Bob knows the true value of $T$, then he would be able to compute  $\mathbb{P}(\ci=I \mid \cf=F)$. To ensure confidentiality, Alice needs to shuffle more times so that $\mathbb{P}(\ci=I \mid \cf=F)$ is close to $\frac{1}{2}$. In the corollary below, we provide lower bounds on the number of shuffles which guarantee that $\mathbb{P}(\ci=I\mid \cf =F)$ is sufficiently close to $\frac{1}{2}$ so as to prevent Bob from guessing Alice's choice. 

\begin{corollary}\label{corollary_bound}Suppose Assumption~\ref{Assumption_equal_probability} holds and biased shuffles characterized through the Markov chain $\{r(t)\}_{t\in \mathbb N_0}$ are repeated $T\in \mathbb{N}_0$ number of times. Then the conditional probability $\mathbb{P}(\ci=I\mid \cf =F)$ satisfies 
\begin{align}
    \left|\mathbb{P}(\ci=I\mid \cf =F) - \frac{1}{2}\right| \leq C. \label{C-result}
\end{align}
with $C\in (0, \frac{1}{2})$, if one of the following conditions hold. \vskip 5pt
Condition 1) Either $T$ is even or $a>b$, and $T \geq \ln(16C/(20-24C))/\ln|a-b|$. \vskip 5pt
Condition 2) $T$ is odd and $a < b$, and  $T \geq \ln(16C/(20+24C))/\ln|a-b|$. 
\end{corollary}

Corollary~\ref{corollary_bound} indicates that if Alice repeats the shuffles sufficiently many times, then it becomes harder for Bob to infer Alice's choice. Proof of Corollary~\ref{corollary_bound} is given in the Appendix.


\section{A More General Bias-Setting: Malicious Shuffling} \label{Sec_Generalization}
For the simplicity of the analysis, we limited our attention to the situation where the value of $\varepsilon$ is positive, which represents the bias caused by the tendencies of Alice's shuffling. However, this limited setting can be further generalized to handle the cases where Bob is a malicious player and tries to make a certain order of cards after the cut more likely.
To reflect such a scenario, we can allow the value of a particular shuffling order $s^*$ to have a probability larger than the rest of other orders by setting $\varepsilon$ to be negative. More specifically, Assumption~\ref{Assumption_nonuniform_prob} can be generalized as follows.

\begin{assumption}\label{Assumption_malicious} The shuffling order satisfies
 \begin{align}
      &\mathbb{P}(s=s^*) =\frac{1}{5}-\varepsilon,\\ 
    &\mathbb{P}(s=j) =\frac{1}{5}+\frac{\varepsilon}{4}\quad\text{for}\,\,j\in \{0,1,2,3,4\}\setminus \{s^*\}, 
  \end{align}
where $\varepsilon\in [-\frac{4}{5},\frac{1}{5}]$ and $s^*\in\{0,1,2,3,4\}$. 
\end{assumption}

Notice that under Assumption~\ref{Assumption_malicious} with $\varepsilon<0$, the shuffling order $s^*$ will have a probability larger than $1/5$, and thus it will be more likely to see this order after shuffling. In this case, the results of Theorem~\ref{theorem_cond_prob} and \ref{Theorem-repeated} can be generalized. For instance, $Case_1$ and $Case_2$ in Theorem~\ref{theorem_cond_prob} can be generalized as 

$Case_1$: For $F \in \{f(I,s^*) : I \in \setInitial\}$ and $I \in \setInitial$,
\begin{align}
  \mathbb{P}(\ci=I \mid \cf=F) =
    \begin{cases}
        \frac{4-20\varepsilon}{8-15\varepsilon}, &\text{if}\,\, F=f(I,s^*) \vspace{0.3cm}\\
        \frac{4+5\varepsilon}{8-15\varepsilon},  &\text{otherwise.}
    \end{cases} 
\end{align}

$Case_2$: For $F \notin \{f(I,s^*) : I \in \setInitial\}$ and $I \in \setInitial$, $\mathbb{P}(\ci=I \mid \cf=F) = \frac{1}{2}$

We note that $Case_1$ and $Case_2$ in Theorem~\ref{Theorem-repeated} can be generalized similarly. Furthermore, similar to the case of a simple random cut, bias in other cyclic shuffling methods such as Hindu cut can be investigated using our methods. 

\begin{remark}[Limitations of the Markov model] When players follow historical patterns in their cuts or use complicated shuffling methods, Markov model with five states may be insufficient and more states may be needed. However, this may result in state-space explosion, and therefore, another model may suit better.  
\end{remark}

\section{Conclusion}\label{Sec_Conclusion}
In this paper, we studied a potential security issue in the Five Card Trick protocol under the setting where there is bias in shuffling. Using the notion of conditional probabilities, we analyzed the likelihood of information leakage and showed that under specific conditions, the confidentiality of a player's choice cannot be fully guaranteed. Furthermore, we extended our analysis to the setting of repeated shuffles. Using a Markov chain model, we gained an insight that repeated shuffles allow players to secure their inputs. Finally, we obtained a lower bound on the number of shuffles required to achieve a desired level of security.

\section*{Acknowledgements}
This work was supported by JSPS KAKENHI Grant No.~JP23K03913.

\section*{Appendix}
We provide proofs of Lemmas~\ref{Lemma-s-star}--\ref{lemma_total_prob} and Corollary~\ref{corollary_bound}. 

\begin{proof}[Lemma~\ref{Lemma-s-star}]
Given $I\in \setInitial$, we define the function $f_I\colon\{0,1,2,3,4\}\to\setFinal$ by $f_I(q)\triangleq f(I,q)$. The equality in \eqref{lem1-result} follows from the fact that 
\begin{align}
    \mathbb{P}\left( f(I,s)=f(I,r)\right)=\mathbb{P}(f_I(s)=f_I(r))
\end{align} and $f_I$ is a one-to-one function. \hfill $\square$
\end{proof}

\begin{proof}[Lemma~\ref{Lemma-s-i}] If $i=j$, then by Lemma~\ref{Lemma-s-star} with $r=0$, we have 
\begin{align} 
\mathbb{P}( f(i,s)=f(j,0)) & = \mathbb{P}( f(j,s)=f(j,0)) = \mathbb{P}(s=0) = \frac{1}{5}-\varepsilon \label{before-applying-assumption}.
\end{align}
Now, assume that $i\neq j$. In this case, we have two options, 1) $i=\re\bl\bl\re\bl, j=\bl\re\bl\re\bl$ or 2) $i=\bl\re\bl\re\bl, j= \re\bl\bl\re\bl$.\\
In option 1), by \eqref{s=j} in Assumption~\ref{Assumption_nonuniform_prob}, we obtain
\begin{align*} 
\mathbb{P}( f(i,s)=f(j,0)) &= \mathbb{P}( f(\re\bl\bl\re\bl,s)=f(\bl\re\bl\bl\re,0))= \mathbb{P}(s=3) = \frac{1}{5}+\frac{\varepsilon}{4}.
\end{align*}
Similarly, in option 2), by using by Lemma~\ref{Lemma-s-star} and \eqref{s=j} in Assumption~\ref{Assumption_nonuniform_prob}, we get 
\begin{align*} 
\mathbb{P}( f(i,s)=f(j,0)) &= \mathbb{P}( f(\bl\re\bl\re\bl,s) = f(\re\bl\bl\re\bl,0)) = \mathbb{P}(s=2) = \frac{1}{5}+\frac{\varepsilon}{4}.
\end{align*}
It follows from the results of both Options 1 and 2 that if $i\neq j$, then 
\begin{align} 
\mathbb{P}( f(i,s)=f(j,0)) & = \frac{1}{5}+\frac{\varepsilon}{4}. \label{second-part}
\end{align}
Finally,~\eqref{before-applying-assumption} and~\eqref{second-part} imply~\eqref{s-i-result}. \hfill $\square$
\end{proof}

\begin{proof}[Lemma~\ref{lemma_total_prob}]
     Let $I=\re\bl\bl\re\bl$ and $\widecheck{I}=\bl\re\bl\re\bl$. Furthermore, let $J\in \{I, \widecheck{I}\}$ be such that $F = f(J,0)$ and $q_{I,F} \in \{1,2,3,4\}$ denotes the index such that $F=f(I,q_{I,F})$ (for instance for $I=\re\bl\bl\re\bl$ and $F=\bl\bl\re\bl\re$, we have $q_{I,F}=4$). By using Law of Total probability, the equality in \eqref{def-cf}, as well as independence of $\ci$ and $s$, we can expand $\mathbb{P}(\cf = F)$ as
    \begin{align}
    \mathbb{P}(\cf = F) &= \mathbb{P}(\ci=I, \cf=F)+\mathbb{P}(\ci=\widecheck{I}, \cf=F) \nonumber\\
    &= \mathbb{P}(\ci=I, f(\ci, s)=F)+\mathbb{P}(\ci=\widecheck{I}, f(\ci,s)=F) \nonumber\\
     &= \mathbb{P}(\ci=I, f(I, s)=F)+\mathbb{P}(\ci=\widecheck{I}, f(\widecheck{I},s)=F) \nonumber\\
          &= \mathbb{P}(\ci=I)\mathbb{P}(f(I,s)=F)+\mathbb{P}(\ci=\widecheck{I})\mathbb{P}(f(\widecheck{I},s)=F). \label{lemma3-1}
    \end{align}
    
    If $F=I$, we have $F = f(I,0)$. In this case,  \eqref{lemma3-1} implies
    \begin{align}
        \mathbb{P}(\cf = F) &=\mathbb{P}(\ci=I)\mathbb{P}(f(I,s)=f(I,0))+\mathbb{P}(\ci=\widecheck{I})\mathbb{P}(f(\widecheck{I},s)=f(I,0)). \label{lemma3-1-1}
    \end{align}
    By Assumption~\ref{Assumption_equal_probability}, we have $\mathbb{P}(\ci=I)=\mathbb{P}(\ci=\widecheck{I})=1/2$. Furthermore, using Lemma~\ref{Lemma-s-i} with $i=j=I$ we obtain $\mathbb{P}(f(I,s)=f(I,0))=1/5-\varepsilon$. Again by using Lemma~\ref{Lemma-s-i} with $i=\widecheck{I}$ and $j=I$, we get $\mathbb{P}(f(\widecheck{I},s)=f(I,0))=1/5+\varepsilon/4$. Therefore, \eqref{lemma3-1-1} implies 
    \begin{align}
        \mathbb{P}(\cf = F) &=\frac{1}{2}(\frac{1}{5}-\varepsilon)+\frac{1}{2}(\frac{1}{5}+\frac{\varepsilon}{4})= \frac{1}{2}(\frac{2}{5}-\frac{3\varepsilon}{4}). \label{Lemma3_J=I}
    \end{align}
    
The case where $F\neq I$ can be handled similarly. In particular, if $F\neq I$, then  $F=\widecheck{I}=f(\widecheck{I}, 0)$. Thus, following the same steps as before, we obtain
  \begin{align}
        \mathbb{P}(\cf = F) &=\mathbb{P}(\ci=I)\mathbb{P}(f(I,s)=f(\widecheck{I},0))+\mathbb{P}(\ci=\widecheck{I})\mathbb{P}(f(\widecheck{I},s)=f(\widecheck{I},0)). \nonumber \\
        &=\frac{1}{2}(\frac{1}{5}-\varepsilon)+\frac{1}{2}(\frac{1}{5}+\frac{\varepsilon}{4})= \frac{1}{2}(\frac{2}{5}-\frac{3\varepsilon}{4}). \label{Lemma3_J=Ic}
    \end{align}
In conclusion,~\eqref{Lemma3_J=I} and~\eqref{Lemma3_J=Ic} confirm~\eqref{lemma_total_prob_eq}. \hfill $\square$
\end{proof}

\begin{proof}[Corollary~\ref{corollary_bound}]
    Consider $Case_2$ in Theorem~\ref{Theorem-repeated}. Since $\mathbb{P}(\ci=I\mid \cf =F)=\frac{1}{2}$, \eqref{C-result} holds for any $T\in \mathbb{N}_0$ regardless of the sign of $a-b$. 
    Now consider $Case_1$. Since 
    \begin{align}
        \frac{4 + 16(a-b)^T}{8+12(a-b)^T} + \frac{4-4(a-b)^T}{8+12(a-b)^T} = 1,
    \end{align}
we have that
    \begin{align}
        \left|\mathbb{P}(\ci=I\mid \cf =F) - \frac{1}{2}\right| &= \left| \frac{4-4(a-b)^T}{8+12(a-b)^T}  - \frac{1}{2}\right| = \left|\frac{20(a-b)^T}{16+24(a-b)^T} \right| \nonumber \\ 
        &= \frac{20|a-b|^T}{|16+24(a-b)^T|} 
    \end{align}
Since $(a-b)\geq \frac{-1}{4}$ (and thus $(a-b)^T\geq \frac{-1}{4}$), it follows that $16+24(a-b)^T\geq 10 > 0$, and therefore, $|16+24(a-b)^T|=16+24(a-b)^T$. This implies  
       \begin{align}
        \left|\mathbb{P}(\ci=I\mid \cf =F) - \frac{1}{2}\right| &= \frac{20|a-b|^T}{16+24(a-b)^T}. \label{p-diff-eq}
    \end{align}
Consider the case where Condition 1 holds. It means, either $T$ is even or $a>b$ holds. In either case, we have $24(a-b)^T=24|a-b|^T$. Therefore, 
    \begin{align}
        \left|\mathbb{P}(\ci=I\mid \cf =F) - \frac{1}{2}\right| &= \frac{20|a-b|^T}{16+24|a-b|^T}. \label{p-diff-ineq-1}
    \end{align}
Since under Condition 1, we have $T\geq \ln(16C/(20-24C))/\ln|a-b|$, noting that $|a-b|<1$ and $ln|a-b|<0$, we obtain 
    \begin{align}
        T \ln|a-b| \leq \ln(16C/(20-24C)). \nonumber
    \end{align} 
  This implies $|a-b|^T \leq 16C/(20-24C)$, 
    and therefore,
    \begin{align}
        20|a-b|^T \leq C(16 +24|a-b|^T).\label{a-b-ineq}
    \end{align} 
    Using \eqref{a-b-ineq} in \eqref{p-diff-ineq-1}, we obtain \eqref{C-result}.

Next, consider the case where Condition 2 holds. In this case, $T$ is odd $a < b$. This implies that 
$24(a-b)^T=-24|a-b|^T$. Therefore, \eqref{p-diff-eq} yields
    \begin{align}
        \left|\mathbb{P}(\ci=I\mid \cf =F) - \frac{1}{2}\right| &= \frac{20|a-b|^T}{16-24|a-b|^T} \label{p-diff-ineq-2}
    \end{align}
Under Condition 1, we have $T\geq \ln(16C/(20+24C))/\ln|a-b|$, noting that $|a-b|<1$ and $ln|a-b|<0$, we obtain 
    \begin{align}
        T \ln|a-b| \leq \ln(16C/(20+24C)). \nonumber
    \end{align} 
  This implies $|a-b|^T \leq 16C/(20+24C)$, and consequently,
    \begin{align}
        20|a-b|^T \leq C(16 -24|a-b|^T). \label{a-b-ineq-2}
    \end{align} 
    Using \eqref{a-b-ineq-2} in \eqref{p-diff-ineq-2}, we obtain \eqref{C-result}, which completes the proof. \hfill $\square$
\end{proof}

\bibliographystyle{plain}
\bibliography{references}
\end{document}